\documentclass[12pt,oneside]{amsart}

\usepackage{amsthm,amsmath,amssymb}
\usepackage{mathrsfs}
\usepackage{enumerate}
\usepackage{amsfonts}
\usepackage{verbatim}
\usepackage{amsbsy}
\usepackage{amsmath}
\usepackage{amssymb}
\usepackage{MnSymbol}
\usepackage{mathrsfs}
\usepackage{xcolor}
\usepackage{cite}

\newtheorem{theorem}{Theorem}[section]
\newtheorem{lemma}[theorem]{Lemma}
\newtheorem{proposition}[theorem]{Proposition}

\newtheorem*{claim}{Claim}

\theoremstyle{definition}
\newtheorem{definition}[theorem]{Definition}

\theoremstyle{remark}

\DeclareMathOperator{\res}{res}

\title{Simulation Of Reaction Systems By The\\ Strictly Minimal Ones}
\author{Wen Chean Teh}
\address{School of Mathematical Sciences\\
	Universiti Sains Malaysia\\
	11800 USM,\linebreak Malaysia}
\email[Corresponding author]{dasmenteh@usm.my}
\author{Adrian Atanasiu}
\address{Consulting Prof. at Faculty of Mathematics and Computer Science\\
	Bucharest University\\
	Str. Academiei 14\\
	Bucharest 010014, Romania}
\email{aadrian@gmail.com}


\begin{document}

\begin{abstract}
Reaction systems, introduced by Ehrenfeucht and Rozenberg, are elementary computational models  based on biochemical reactions transpiring within the living cells. Numerous studies focus on mathematical aspects of minimal reaction systems due to their simplicity and rich generative power. In 2014 Manzoni, Pocas, and Porreca showed that every reaction system can be simulated by some minimal reaction system over an extended background set. Motivated by their work, we  introduce the concepts of strictly minimal and hybrid reaction systems. Using our new concepts, the result of Manzoni et al.~is revisited and strengthened.
We also show that extension of the background set by polynomially bounded many elements is not sufficient to guarantee the aforementioned simulation. Finally, an analogous result for strong simulation is obtained.
\end{abstract}

\maketitle
\section{Introduction}

Reaction systems, introduced in 2007 by Ehrenfeucht and Rozenberg \cite{ehrenfeucht2007reaction}, are elementary computational models  inspired 
by biochemical reactions taking place within the living cells. 
This study belongs to one of the diverse research lines initiated in \cite{ehrenfeucht2011functions} that pertains to mathematical study of state transition functions specified by reaction systems, called $rs$ functions. For a motivational survey on reaction systems, we refer the reader to Ehrenfeucht, Petre, and Rozenberg \cite{ehrenfeucht2017reaction}.

Minimal reaction systems \cite{ehrenfeucht2012minimal}, where the number of resources in each reaction is minimal, have been relatively well-studied due to their simplicity.
Salomaa \cite{salomaa2014compositions} initiated the study on the generative power under composition of $rs$ functions specified by minimal reaction systems and later we showed that not every $rs$ function can be thus generated for the quarternary alphabet  \cite{teh2018compositions}. On the other hand, 
Manzoni, Pocas, and Porreca \cite{manzoni2014simple} 
introduced the study of simulation by reaction systems and they showed that every $rs$ function  can be simulated by some minimal reaction system over an extended background set.
Other studies on mathematical properties of minimal reaction systems include \cite{azimi2017steady, salomaa2014minimal,  salomaa2017minimal,  teh2017minimal}.

This study refines and expands on the study of simulation by reaction systems initiated in \cite{manzoni2014simple}. We propose strictly minimal reaction systems as a 
new canonical class of reaction systems. We also introduce hybrid reaction systems, where the reactant set, inhibitor set, and product set of each reaction are allowed to contain entities from different background sets. Then the result by Manzoni et al.~is revisited and strengthened  by showing that the extended background set can be fixed ahead independent of the given $rs$ function. Next, we show that the number of extra resources needed in the fixed extended background set cannot be bounded polynomially in terms of the size of the original background set. Finally, a stronger version of simulation is studied and it will be shown that minimal reaction systems are in fact rich enough to strongly simulate every $rs$ function over a given background set.

\section{Preliminaries}

If $S$ is any finite set, then the cardinality of $S$ is denoted by $\vert S\vert$ and the power set of $S$ is denoted by $2^S$.

From now onwards, unless stated otherwise, $S$ is a fixed finite nonempty set.

\begin{definition}
	A \emph{reaction in $S$} is  a triple
	$a=(R_a,I_a, P_a)$, where $R_a$ and $I_a$ are (possibly empty) disjoint subsets of $S$ and $P_a$ is a nonempty subset of $S$. The sets $R_a$, $I_a$, and $P_a$ are the \emph{reactant set}, \emph{inhibitor set}, and \emph{product set} respectively. The pair $(R_a, I_a)$ is the \emph{core of $a$}.
	\end{definition}

\begin{definition}\label{2606a}
	A \emph{reaction system over $S$} is a pair $\mathcal{A}=(S,A)$ where $S$ is called the \emph{background set} and $A$ is a (possibly empty) set of reactions in $S$. We say that $\mathcal{A}$ is \emph{nondegenerate} if $R_a$ and $I_a$ are both nonempty for every $a\in A$ and  $\mathcal{A}$ is \emph{maximally inhibited} if $I_a= S\backslash R_a$ for every $a\in A$.
\end{definition}

\begin{definition}
	Suppose \mbox{$\mathcal{A}=(S,A)$} is a reaction system. The \emph{state transition function} 
	$\res_{\mathcal{A}} \colon 2^S\rightarrow 2^S$  is defined by
	$$\res_{\mathcal{A}}(X)= \bigcup  \{\, P_a \mid  a\in A \text{ such that }
R_a\subseteq X\text{ and } I_a\cap X=\emptyset\,\},\quad    \text{for all }X\subseteq S.$$
\end{definition}

If $R_a\subseteq X$ and $I_a\cap X=\emptyset$, we say that the reaction $a$ is \emph{enabled by $X$}. Hence, $\res_{\mathcal{A}}(X)$ is the cumulative union of product sets of all reactions enabled by $X$. For our purpose and without loss of generality, we may assume that
distinct reactions in $A$ do not have the same core.

\begin{definition}\label{311017b}
	Every function $f\colon 2^S\rightarrow 2^S$ is called an \emph{$rs$ function over $S$}.
We say that $f$ can be \emph{specified by a reaction system} $\mathcal{A}$ over $S$ if $f=\res_{\mathcal{A}}$.
\end{definition}

Since every $rs$ function over $S$ can be canonically specified by a unique maximally inhibited reaction system over $S$, it follows that the class of $rs$ functions over $S$ is exactly the class of state transition functions over $S$.

\begin{definition}\cite{ehrenfeucht2012minimal, teh2017minimal} \label{1507d}
	Suppose $\mathcal{A}=(S,A)$ is reaction system.
	Then  $\mathcal{A}$ is \emph{minimal} if $\vert R_a\vert\leq 1$
	and $\vert I_a\vert\leq 1$ for every reaction $a\in A$.
\end{definition}

The elements in the reactant set or inhibitor set of a reaction $a$ are called the \emph{resources} of $a$. The classification of any reaction system according to the total number of resources allowed in each of its reactions was initiated by Ehrenfeucht, Main, and Rozenberg \cite{ehrenfeucht2011functions}. 
From time to time, nondegeneracy has been a naturally adopted convention. Hence, every minimal reaction system satisfies
$\vert R_a\vert=\vert I_a\vert=1$ for each $a\in A$ in the early studies.
A characterization of $rs$ functions that can be specified by minimal reaction systems was obtained by Ehrenfeucht, Kleijn, Koutny, and Rozenberg \cite{ehrenfeucht2012minimal}.
 Later the same characterization was extended in \cite{teh2017minimal} to cover for degenerate reaction systems as well.
We present this characterization due to its historical significance.

\begin{theorem} \cite{ehrenfeucht2012minimal, teh2017minimal}    \label{2305b}
	Suppose $f$ is an $rs$ function over $S$. 
	Then $f=\res_\mathcal{A}$ for some (possibly degenerate) minimal reaction system  $\mathcal{A}$ if and only if $f$ satisfies the following two properties:
	\begin{itemize}
	\item (Union-subadditivity) $f(X\cup Y)\subseteq f(X)\cup f(Y)$ for all $X,Y\subseteq S$;
	\item (Intersection-subadditivity) $f(X\cap Y)\subseteq f(X)\cup f(Y)$ for all $X,Y\subseteq S$.
	\end{itemize}
\end{theorem}

 The following definition of simulation was introduced by Manzoni et al.~\cite{manzoni2014simple}.

\begin{definition}
	Suppose $f$ is an $rs$ function over $S$ and $k$ is a positive integer. Suppose $S\subseteq S'$ and $\mathcal{A}$ is a reaction system over $S'$. We say that $f$ can be \emph{\mbox{$k$-simulated} by} $\mathcal{A}$    if for every $X\subseteq S$, 
	$$f^n(X)=\res_{\mathcal{A}}^{kn}(X)\cap S\quad \text{ for all positive integers } n.$$ 
\end{definition}


The following observation says that $1$-simulation do not 
add to the expressive power of reaction systems.

\begin{proposition}\label{230118b}
	Suppose $f$ is an $rs$ function over $S$ and $S\subseteq S'$. Suppose $f$ can be $1$-simulated
	by some reaction system $\mathcal{A}'=(S',A')$ over $S'$. Then $f$ can be specified by the reaction system $\mathcal{A}=(S,A)$ over $S$, where
	$$A=\{  \,(R_a, I_a \cap S, P_a\cap S)\mid a\in A'\text{ and } R_a\subseteq S    \, \}.$$
\end{proposition}

\begin{proof}
The definition of $1$-simulation implies that $f(X)= \res_{\mathcal{A}'}(X)\cap S$ for every $X\subseteq S$. Fix an arbitrary $X\subseteq S$.
	It suffices to show that $\res_{\mathcal{A}}(X)= \res_{\mathcal{A}'}(X)\cap S$.
	Suppose $x\in \res_{\mathcal{A}}(X)$. Then  
	$(R_a, I_a \cap S, P_a\cap S)$ is enabled by $X$ for some $a\in A'$ with $x\in P_a \cap S$.
It follows that $a=(R_a, I_a, P_a)$ is enabled by $X$ because $X\subseteq S$ and thus $x\in \res_{\mathcal{A}'}(X)\cap S$. Conversely, if 
$x\in \res_{\mathcal{A}'}(X)\cap S$. Then 	
	$a=(R_a, I_a, P_a)$ is enabled by $X$ for some $a\in A'$ with $x\in P_a$.
	Hence, $(R_a, I_a \cap S, P_a\cap S)\in A$ 
	is enabled by $X$. Since $x\in P_a\cap S$, it follows that $x\in \res_{\mathcal{A}}(X)$.
\end{proof}

 Manzoni~et al.~\cite{manzoni2014simple}  showed that minimal reaction systems are rich enough for the purpose of simulation. Their result serves as the main motivation for this study. We observe that the 
number of resources in each reaction of the
reaction system constructed in their proof is actually one. Therefore, we introduce the following definition before stating what they have actually shown.

\begin{definition} \label{030118a}
	Suppose $\mathcal{A}=(S,A)$ is reaction system.
	Then  $\mathcal{A}$ is \emph{strictly minimal} if $\vert R_a \cup I_a\vert\leq 1$
	for every reaction $a\in A$.
\end{definition}

\begin{theorem}\cite{manzoni2014simple}\label{1107a}
	Suppose $f$ is an $rs$ function over $S$. Then there exists a strictly minimal reaction system $\mathcal{B}$ over some $S'\supseteq S$ such that $f$ can be \mbox{$2$-simulated} by $\mathcal{B}$.
\end{theorem}

\section{Hybrid Reaction Systems}

There are studies on mathematical properties of reaction systems, for example, the totalness of state transition functions and the functional completeness of the reaction systems as in Salomaa~\cite{salomaa2012functions},
 where the properties do not depend on the nonempty product sets. More importantly, the
 reactant set, inhibitor set, and product set of each reaction in the reaction system constructed in the proof of Theorem~\ref{1107a} appear to contain entities of different nature. 
 These observations motivate our definition of hybrid reaction system, where the output elements are allowed to come from a different background set whenever a reaction is enabled.

\begin{definition}
Suppose $S$ and $T$ are finite nonempty sets. An \emph{$(S,T)$-reaction} is a triple of sets $a=(R_a, I_a, P_a)$ such that $R_a$ and $I_a$ are (possibly empty) disjoint subsets of $S$ and $P_a$ is a nonempty subset of $T$. A \emph{hybrid reaction system} \emph{over $(S,T)$} is a triple $\mathcal{A}=(S,T,A)$ where $S$ and $T$ are the \emph{background sets} and $A$ is a (possibly empty) set of $(S,T)$-reactions.
\end{definition}

Obviously, a hybrid reaction system over $(S,T)$ becomes a (normal) reaction system when $S=T$. Basic terminology of reaction systems carries over to hybrid reaction systems analogously. Hence, the reader is assumed to know, for example, the definition of the state transition function $\res_{\mathcal{A}}$  and what it means for $\mathcal{A}$ to be maximally inhibited when $\mathcal{A}$ is a hybrid reaction system.   Furthermore, every $rs$ function
$f\colon 2^S\rightarrow 2^T$ can be canonically specified by a unique maximally inhibited hybrid reaction system over $(S,T)$.


The following theorem says that every reaction system can be naturally decomposed into two strictly minimal hybrid reaction systems. This theorem is essentially extracted from the proof of Theorem~\ref{1107a}. 

\begin{theorem}\label{0403a}
Suppose $\mathcal{A}=(S,A)$ is a reaction system. Let $T=\{ \,\bar{a}\mid a\in A\,\}$ where $\bar{a}$ is a distinguished symbol for each $a\in A$.
Let $$C=\{\,(\emptyset, \{x\}, \{\bar{a}\}) \mid a\in A \text{ and } x\in R_a\,\}
\cup \{\,(\{y\}, \emptyset, \{\bar{a}\})\mid a\in A \text{ and } y\in I_a\,\}$$
and $D= \{\,(\emptyset, \{\bar{a}\}, P_a)\mid a\in A \,\}$.
Then $\mathcal{C}=(S,T,C)$ and $\mathcal{D}=(T,S,D)$ are strictly minimal hybrid reaction systems such that $\res_{\mathcal{A}}=\res_{\mathcal{D}}\circ \res_{\mathcal{C}}$.
\end{theorem}

\begin{proof}
Note that $\res_{\mathcal{D}}(Y) = \bigcup\{ \,P_a\mid \bar{a}\in T\backslash Y\,\}$ for every $Y\subseteq T$. Therefore, it suffices to show that
 $$\res_{\mathcal{C}}(X)=\{\,\bar{a}\in T\mid a \text{ is not enabled by } X\,\} \quad\text{ for all } X\subseteq S$$
because $\res_{\mathcal{A}}=\res_{\mathcal{D}}\circ \res_{\mathcal{C}}$ would then follow immediately.

Suppose $X\subseteq S$ and $a\in A$. By definition, $a$ is not enabled by $X$ if and only if $R_a\nsubseteq X$ or $I_a\cap X\neq \emptyset$. If $x\in R_a\backslash  X$,  then $(\emptyset,\{x\}, \{\bar{a}\})\in C$ is enabled by $X$.
Similarly, if $y\in I_a\cap X$, then $(\{y\},\emptyset, \{\bar{a}\})\in C$ is enabled by $X$. It follows that $\bar{a}\in \res_{\mathcal{C}}(X)$ whenever $a$ is not enabled by $X$.

Conversely, suppose 
$\bar{a}\in \res_{\mathcal{C}}(X)$. Then 
some $c\in C$ such that $P_c=\{\bar{a}\}$ is enabled by $X$. If $c=(\emptyset,\{x\}, \{\bar{a}\})$ for some $x\in R_a$, then $x\in R_a\backslash X$ and so $R_a\nsubseteq X$. Similarly, if $c =(\{y\},\emptyset, \{\bar{a}\})$ for some $y\in I_a$, then 
$y\in I_a\cap X$ and so $I_a\cap X\neq \emptyset$. It follows that  $a$ is not enabled by $X$ whenever 
$\bar{a}\in \res_{\mathcal{C}}(X)$.
\end{proof}

In view of the proof of Theorem~\ref{0403a}, it is intriguing whether there are  
$\mathcal{C}$ and $\mathcal{D}$ 
such that $\res_{\mathcal{C}}(X)$ is the set of $\bar{a}$ such that $a$ is enabled by $X$ and $\res_{\mathcal{A}}=\res_{\mathcal{D}}\circ \res_{\mathcal{C}}$. 
Trivially, we can take $C=\{\,(R_a, I_a, \{\bar{a}\}) \mid a\in A \,\}$
and $D= \{\,(\{\bar{a}\},\emptyset, P_a)\mid a\in A \,\}$. 
However, it will only be interesting if such $\mathcal{C}$ exists where its complexity is less than 
$\mathcal{A}$. By our next claim, this is not possible.

\begin{claim}
Suppose $\mathcal{C}=(S,T,C)$ is any hybrid reaction system such that $\res_{\mathcal{C}}(X)$ is the set of $\bar{a}$ such that $a$ is enabled by $X$ for each $X\subseteq S$. Then for every reaction 
$c\in C$, if $\bar{a}\in P_c$, then $R_a\subseteq R_c$ and $I_a\subseteq I_c$.
\end{claim}

\begin{proof}
Suppose $c\in C$ and $\bar{a}\in P_c$. Clearly, $c$ is enabled by $R_c$ and thus $\bar{a}\in \res_{\mathcal{C}}(R_c)$. By the hypothesis, $a$ is enabled by $R_c$, implying that $R_a\subseteq R_c$. Similarly, $c$ is enabled by $S\backslash I_c$ and thus $\bar{a}\in \res_{\mathcal{C}}(S\backslash I_c)$. By the hypothesis again, $a$ is enabled by $S\backslash I_c$, implying that $(S\backslash I_c) \cap I_a=\emptyset$ and thus $I_a\subseteq I_c$.	
\end{proof}

Theorem~\ref{0403a} justifies the canonicalness of strictly minimal hybrid reaction systems, as functions specified by them  can generate every $rs$ function $f$ over $S$ under composition. However, the hybrid reaction system $\mathcal{C}$ as in the theorem depends on $\mathcal{A}$ such that $\res_{\mathcal{A}}=f$.
Therefore, the next theorem is a variation of Theorem~\ref{0403a} where the hybrid reaction system $\mathcal{C}$ is independent from $f$.
This theorem is essentially implied by the proof  of Theorem~4 in Salomaa~\cite{salomaa2015two}, although over there any reaction system is required to be nondegenerate. The main idea is to give a name to each subset of the background set.
An alternative original proof of this next theorem can be found in \cite{teh2018note}. 

\begin{theorem}\label{0407a}
Let $T=\{\,N_X \mid X\subseteq S          \,\}$ where
$N_X$ is a distinguished symbol for each $X\subseteq S$. 
Let
$$C= \{\,( \emptyset, \{x\},\{ N_X\})\mid X\subseteq S \text{ and }  x\in X   \,\} \cup   \{\,( \{y\},\emptyset , \{N_X \} )\mid X\subseteq S\text{ and } y\in S\backslash X   \,\}.
$$ 
Suppose $f$ is an $rs$ function over $S$.
Let $$D=\{\,(\emptyset, \{N_X\}, f(X))\mid X\subseteq S \text{ and } f(X)\neq \emptyset     \,\}. $$ Then $\mathcal{C}=(S, T, C)$ and $\mathcal{D}=(T,S,D)$ are strictly minimal hybrid reaction systems such that $\res_{\mathcal{D}} \circ \res_{\mathcal{C}}= f$.
\end{theorem} 

\begin{proof}
Let $\mathcal{A}=(S,A)$ be the canonical maximally inhibited reaction system such that $f=\res_{\mathcal{A}}$, that is,
where $$A=\{\,(X,S\backslash X, f(X)) \mid X\subseteq S \text{ and } f(X)\neq \emptyset    \,\}.$$
Every $X\subseteq S$ can be uniquely associated to the reaction $a_X=(X,S\backslash X, f(X))$. Let $N_X$ be the distinguished symbol $\overline{a_X}$ for each $X\subseteq S$. Then it can be verified that
$$C=\{\,(\emptyset, \{x\}, \{\bar{a}\}) \mid a\in A \text{ and } x\in R_a\,\}
\cup \{\,(\{y\}, \emptyset, \{\bar{a}\})\mid a\in A \text{ and } y\in I_a\,\}$$
and $D= \{\,(\emptyset, \{\bar{a}\}, P_a)\mid a\in A \,\}$. Therefore,
by Theorem~\ref{0403a}, it follows that $\res_{\mathcal{D}} \circ \res_{\mathcal{C}}=
\res_{\mathcal{A}}=f$.
\end{proof}

Using Theorem~\ref{0407a} and adapting the proof of Theorem~\ref{1107a}, we now strengthen Theorem~\ref{1107a} by showing that the extended background set for the simulating reaction system can be chosen ahead independent from the given $rs$ function. Before that, we need a lemma.

\begin{lemma}\label{2612a}
Suppose $\mathcal{C}=(S,S',C)$ and $\mathcal{D}=(T,T',D)$ are hybrid reaction systems. Let $\mathcal{A}$ be the hybrid reaction system  $(S\cup T, S' \cup T', C\cup D)$. Then $$\res_{\mathcal{A}}(X)=\res_{\mathcal{C}}(X\cap S)\cup \res_{\mathcal{D}}(X\cap T), \quad \text{ for all } X\subseteq S\cup T.$$
\end{lemma}

\begin{proof}
Suppose $X\subseteq S\cup T$. Then 
\begin{align*}
\res_{\mathcal{A}}(X)
={} & \bigcup_{\substack{c\in C\\
R_c\subseteq X, I_c\cap X=\emptyset} } P_c   \quad  \cup  \quad \bigcup_{\substack{d\in D\\
R_{d}\subseteq X, I_{d}\cap X=\emptyset} } P_{d} \\
={} & \bigcup_{\substack{c\in C\\
R_c\subseteq X \cap S, I_c\cap (X\cap S)=\emptyset} } P_c \quad\cup \quad \bigcup_{\substack{d\in D\\
R_{d}\subseteq X\cap T, I_{d}\cap (X\cap T)=\emptyset} } P_{d} \\
={}& \res_{\mathcal{C}}(X\cap S)\cup \res_{\mathcal{D}}(X\cap T).
\end{align*}
\end{proof}

\begin{theorem}\label{0407b}
	There exists a fixed set $S'\supseteq S$ such that every $rs$ function over $S$ can be $2$-simulated by some strictly minimal reaction system  over $S'$.
\end{theorem}

\begin{proof}
Let $T=\{\,N_X\mid X\subseteq S\,\}$ as in Theorem~\ref{0407a} and $S'=S\cup T$.
Then $S'$ is a fixed background set extending $S$.
Suppose $f$ is an $rs$ function over $S$.
Let $\mathcal{C}$ and $\mathcal{D}$ be as in Theorem~\ref{0407a}.
Consider the reaction system $\mathcal{A}=(S', C\cup D)$ over $S'$.
Clearly, $\mathcal{A}$ is strictly minimal.

\begin{claim}
For every integer $n\geq 2$ and every $X\subseteq S'$,
$$\res^{n}_{\mathcal{A}}(X)\cap S=( \res_{\mathcal{D}}\circ  \res_{\mathcal{C}}) (\res_{\mathcal{A}}^{n-2}(X)\cap S),$$
where $\res^{0}_{\mathcal{A}}(X)=X$.
\end{claim}

\begin{proof}[Proof of the claim]
We argue by mathematical induction.
For the base step, suppose $X\subseteq S'$.
By Lemma~\ref{2612a},
$\res_{\mathcal{A}}(X)=\res_{\mathcal{C}}(X\cap S)\cup \res_{\mathcal{D}}(X\cap T)$
and  
 $$\res_{\mathcal{A}}^2(X)= \res_{\mathcal{A}}(\res_{\mathcal{A}}(X))= \res_{\mathcal{C}}( \res_{\mathcal{A}}(X)  \cap S)\cup \res_{\mathcal{D}}( \res_{\mathcal{A}}(X)  \cap T).$$
Since $\res_{\mathcal{C}}\colon 2^{S}\rightarrow 2^T$,
$\res_{\mathcal{D}}\colon 2^{T}\rightarrow 2^S$, and $S\cap T=\emptyset$, it follows that
$$\res_{\mathcal{A}}^2(X)\cap S
=\res_{\mathcal{D}}( \res_{\mathcal{A}}(X)  \cap T)
= \res_{\mathcal{D}}(\res_{\mathcal{C}}(X\cap S))  = ( \res_{\mathcal{D}}\circ  \res_{\mathcal{C}})(X\cap S).$$
Thus the base step is complete.

For the induction step, suppose $X\subseteq S'$. 
Then $$\res_{\mathcal{A}}^{n+1}(X)\cap S
= \res_{\mathcal{A}}^{2}( \res_{\mathcal{A}}^{n-1}(X)  ) \cap S
= (\res_{\mathcal{D}}\circ \res_{\mathcal{C}}) (\res_{\mathcal{A}}^{n-1}(X)\cap S)  .$$
The last equality follows from the base step. The induction step is complete. \renewcommand{\qedsymbol}{}
\end{proof}

The fact that $f$ can be $2$-simulated by $\mathcal{A}$ follows immediately from the next claim.

\begin{claim}
For every positive integer $n$ and every $X\subseteq S'$,
$$f^n(X\cap S)=  \res^{2n}_{\mathcal{A}}(X)\cap S.$$
\end{claim}

\begin{proof}[Proof of the claim]
We argue by mathematical induction. 
By Theorem~\ref{0407a} and the previous claim, $\res^2_{\mathcal{A}}(X)\cap S=(\res_{\mathcal{D}}\circ \res_{\mathcal{C}})(X\cap S)=f(X\cap S)$ for all $X\subseteq S'$,
thus the base step is done.
For the induction step, suppose $X\subseteq S'$.
By the previous claim again, $\res^{2(n+1)}_{\mathcal{A}}(X)\cap S
=( \res_{\mathcal{D}}\circ  \res_{\mathcal{C}}) (\res_{\mathcal{A}}^{2n}(X)\cap S) $, which equals $ f ( f^n(X\cap S)  )
=f^{n+1}(X\cap S)$ by Theorem~\ref{0407a} and the induction hypothesis. \renewcommand{\qedsymbol}{}
\end{proof}
Therefore, the proof is complete.
\end{proof}

\section{Extension of Background Set is Necessary}

In view of Theorem~\ref{0407b}, the following theorem says that $S'$ needs to properly extend $S$ if every $rs$ function over $S$ is to be $k$-simulated for some positive integer $k$ by some strictly minimal reaction system over $S'$.

\begin{theorem}\cite{manzoni2014simple}\label{140118a}
	Suppose $\vert S\vert\geq 3$.
	There exists an $rs$ function over $S$ that cannot be $k$-simulated for any positive integer $k$ by any minimal reaction system over $S$.
\end{theorem}

The next lemma is a generalization of what was actually shown by the first half of the proof of Theorem~\ref{140118a}.

\begin{lemma}\label{060118a}
	Suppose $\vert S\vert=n\geq 2$ and  let $X_1, X_2, \dotsc, X_{2^n}$ be any enumeration of all the subsets of $S$.
		Let $f$ be the $rs$ function over $S$ defined by
		$$f(X_i)=\begin{cases}
		X_{i+1} &\text{ if } 1\leq i<2^n\\
		X_{2^n} &\text{ if } i= 2^n.
		\end{cases}$$
Suppose $S'$ is a finite background set extending $S$.
Then $f$ cannot be $k$-simulated by any reaction system over $S'$ whenever   $k > \frac{2^{\vert S'\vert}-2}{2^{\vert S\vert}-2  }$.
\end{lemma}

\begin{proof}
We argue by contradiction. Suppose $\mathcal{A}$ is an arbitrary reaction system over $S'$ and $k$ is an arbitrary integer such that $k > \frac{2^{\vert S'\vert}-2}{2^{\vert S\vert}-2  }$. Assume $f$ can be $k$-simulated by $\mathcal{A}$.
	Consider the state sequence
	$X_1, \res_{\mathcal{A}}(X_1), \res^2_{\mathcal{A}}(X_1),\dotsc$
	 generated by $\mathcal{A}$ with initial state $X_1$.
	 Since $k\cdot (2^{n} -2)+2 > 2^{\vert S'\vert}$, the following initial terms
	$$X_1, \res_{\mathcal{A}}(X_1), \res^2_{\mathcal{A}}(X_1),  \dotsc,\res^{k(2^n-2)+1}_{\mathcal{A}}(X_1) $$  
	cannot be all distinct subsets of $S'$. Hence, 
	$\res^{k(2^n-2)}_{\mathcal{A}}(X_1)$ is part of a cycle, say with period $p\geq 1$. Then 
$$	\res^{k(2^n-2+p)}_{\mathcal{A}}(X_1)=\res^{kp}_{\mathcal{A}}(\res^{k(2^n-2)}_{\mathcal{A}}(X_1))= \res^{k(2^n-2)}_{\mathcal{A}}(X_1)$$
and thus 
$$\res^{k(2^n-2+p)}_{\mathcal{A}}(X_1)\cap S=	
\res^{k(2^n-2)}_{\mathcal{A}}(X_1)\cap S=
f^{2^n-2}(X_1)=X_{2^n-1}.$$

	However, $\res^{k(2^n-2+p)}_{\mathcal{A}}(X_1)\cap S= f^{2^n-2+p}(X_1)=X_{2^n}$, which is a contradiction.	
\end{proof}

The conclusion of Theorem~\ref{140118a}  is true for $\vert S\vert=2$ as well.
Let $S=\{a,b\}$ and consider the $rs$ function $f$ defined by $f(\{a\})=\{b\}$, $f( \{b\})=\emptyset$, $f(\emptyset)=S$, and $f(S)=S$.
By Lemma~\ref{060118a}, $f$ cannot be $k$-simulated by any reaction system over $S$ whenever $k>1$.
Since $f(S)\not\subseteq f(\{a\})\cup f( \{b\} )$, it follows that $f$ is not union-subadditive and thus cannot be $1$-simulated (equivalently, specified) by any minimal reaction system over $S$ by Theorem~\ref{2305b}.

In view of Theorem~\ref{0407b}, we now show that when the background set $S$ is extended by a fixed finite number of elements, it is not generally sufficient to \mbox{$2$-simulate} every $rs$ function over $S$ by some strictly minimal reaction systems over the extended background set. In fact, the following much stronger statement holds.

\begin{theorem}\label{120219}
No polynomial $p$ has the property that 
for every set $S$, there exists a set
$S'\supseteq S$ with $\vert S'\vert \leq p(\vert S\vert)$ such that every $rs$ function over $S$ can be $k$-simulated  by some strictly minimal reaction system over $S'$ for some positive integer $k$.
\end{theorem}

\begin{proof}
First of all, notice that the $rs$ function $f$ as defined in
Lemma~\ref{060118a} is uniquely determined by the corresponding sequence $X_1, X_2, \dotsc, X_{2^{\vert S\vert}}$. Hence, by some simple counting argument, there are $2^{\vert S\vert}!$ distinct such  $rs$ functions over $S$. On the other  hand, there are $(2^{\vert S'\vert})^{2\vert S'\vert+1}$ distinct strictly minimal reaction systems over $S'$ because there are $2\vert S'\vert+1$ distinct possible cores (which include $(\emptyset, \emptyset)$)
for reactions in such reaction systems.


Fix a polynomial $p$.
Suppose $S\subseteq S'$ and $\vert S'\vert \leq p(\vert S\vert)$. 
Since $2^{p(\vert S\vert)}\geq  2^{\vert S'\vert}  > \frac{2^{\vert S'\vert}-2}{2^{\vert S\vert}-2  }$ (for $\vert S\vert\geq 2$),
by Lemma~\ref{060118a}, 
any such $f$ cannot be $k$-simulated by any reaction system over $S'$ whenever   $k > 2^{p(\vert S\vert)}$ (in fact, whenever $k>\frac{2^{\vert S'\vert}-2}{2^{\vert S\vert}-2  }$).
Meanwhile, for every positive integer $k$, by definition, there are at most $(2^{p(\vert S\vert)})^{2p(\vert S\vert)+1  }$ $rs$ functions over $S$ that can be \mbox{$k$-simulated} by some strictly minimal reaction system over $S'$.
Therefore, at most $2^{p(\vert S\vert)} \cdot
(2^{p(\vert S\vert)})^{2p(\vert S\vert)+1 }$ $rs$ functions over $S$ can be $k$-simulated by some strictly minimal reaction system over $S'$ for some $k\leq 2^{p(\vert S\vert)}$. When $\vert S\vert$ is sufficiently large, 
$$2^{p(\vert S\vert)} \cdot
(2^{p(\vert S\vert)})^{2p(\vert S\vert)+1 }=2^{p(\vert S\vert) ( 2p(\vert S\vert)+2   ) }
<2^{\vert S\vert}!,$$
and thus it follows that  not all of the $2^{\vert S\vert}!$ $rs$ functions over $S$ as defined in Lemma~\ref{060118a}
can be $k$-simulated  by some strictly minimal reaction system over $S'$ for some positive integer $k$.
\end{proof}

\section{Strong $k$-Simulation by Reaction Systems}

Now, we formally study a stronger version of $k$-simulation, which was first proposed in the conclusion section of \cite{manzoni2014simple}.

\begin{definition}
Suppose $f$ is an $rs$ function over $S$ and $k$ is a positive integer. Suppose $S\subseteq S'$ and $\mathcal{A}$ is a reaction system over $S'$. 
We say that $f$ can be \emph{strongly $k$-simulated by $\mathcal{A}$} if{f} 
$f(X)=\res_{\mathcal{A}}^{k}(X)$ for all $X\subseteq S$.
\end{definition}

Some direct computation shows that the $2$-simulating strictly minimal reaction system constructed in the proof of Theorem~\ref{0407b} generally does not strongly \mbox{$2$-simulate} the given $rs$ function.

The following is a reformulation of Theorem~3 in \cite{salomaa2015two}, which is a strong version analogue of Theorem~\ref{1107a}.

\begin{theorem}\cite{salomaa2015two}\label{1407a}
	Suppose $f$ is an $rs$ function over $S$ such that $f(\emptyset)=\emptyset$. Then there exists a minimal reaction system $\mathcal{B}$ over some $S'\supseteq S$ such that $f$ can be strongly $2$-simulated by $\mathcal{B}$.
\end{theorem}

Salomaa adopted the convention that every reaction system is nondegenerate. Relaxing this constraint, we strengthen Theorem~\ref{1407a}, not only by having a fixed extended background set independent from the given $rs$ function, but also by generalizing it to every $rs$ function over $S$. First, we need  a technical lemma, which is an adaptation of Theorem~\ref{0407a} to suit our purpose.

\begin{lemma}\label{0707a}
	Let $T=\{\,N_X \mid \emptyset \neq X\subseteq S          \,\}\cup \{\ast, \diamond\}$, where $N_X$ is a distinguished symbol for each $\emptyset \neq X\subseteq S$.
Let
\begin{multline*}
C=\{\,(\{y\}, \emptyset, \{N_X \})\mid \emptyset \neq X\subseteq S \text{ and } y\in S\backslash X\,\} \,\cup \,	\{\,(\{s\},\emptyset, \{\ast\} )\mid s\in S\,\} \,\, \cup \\
		\{\,(\{x\}, \{x'\}, \{N_X\}) \mid X\subseteq S \text{ and } x,x'\in X \text{ with } x\neq x'\,\} \,\cup\, \{ ( \emptyset, \{\diamond\}, \{ \diamond  \}       )   \}.
 	\end{multline*}
Suppose $f$ is an $rs$ function over $S$. Let 
$$D= \{\,(\{\ast\}, \{N_X\}, f(X))\mid \emptyset \neq X\subseteq S \text{ and } f(X)\neq \emptyset \,\}\,\cup\, \{ ( \{ \diamond  \}, \{ \ast\}, f(\emptyset) )     \}.$$
Then $\mathcal{C}=(S\cup \{\diamond  \}, T, C)$ and $\mathcal{D}=(T,S,D)$ are hybrid minimal reaction systems such that
$(\res_{\mathcal{D}} \circ \res_{\mathcal{C}})(X)= f(X)$ for all $X\subseteq S$.
\end{lemma}

\begin{proof}
Note that $\res_{\mathcal{D}} (\res_{\mathcal{C}}(\emptyset))=
\res_{\mathcal{D}} (\{ \diamond  \})= f(\emptyset)$.
Suppose $X$ is a nonempty subset of $S$.
It suffices to show that
	$\res_{\mathcal{C}}(X)= T \backslash \{N_X\} $ because then 
	only the reaction $(\{\ast\}, \{N_X\}, f(X))$
	is enabled  by $\res_{\mathcal{C}}(X)$,  provided $f(X)\neq \emptyset$,
	and thus $\res_{\mathcal{D}} (\res_{\mathcal{C}}(X))=f(X)$.  
	
	Suppose  $Y$ is a nonempty subset of $S$ distinct from $X$. 
	If $X\nsubseteq Y$, say $x\in  X\backslash Y$, then the reaction 
	$( \{x\},\emptyset , \{N_Y\} )$ is enabled by $X$. Otherwise, if $X\subseteq Y$ and so $Y\nsubseteq X$, then the reaction 
	$( \{y\},\{y'\}, \{N_Y\} )$ is enabled by $X$ for any $y\in X \,(\neq \emptyset)$ and $y'\in  Y\backslash X$.	In either case, $N_Y\in \res_{\mathcal{C}}(X)$.
	
	On the other hand, none of the reactions in $C$ with product set being $\{N_X\}$ is enabled by $X$. Furthermore, $\{ \ast, \diamond \}\subseteq  \res_\mathcal{C}(X)$ because $X$ is a nonempty subset of $S$. 
Therefore, $\res_{\mathcal{C}}(X)= T\backslash \{N_X\}$  as required.
\end{proof}

\begin{theorem}\label{0603a}
There exists a fixed set $S'\supseteq S$ such that 
every $rs$ function $f$ over $S$ can be strongly $2$-simulated by some minimal reaction system over $S'$. 
\end{theorem}

\begin{proof}
Let $T=\{\,N_X \mid \emptyset \neq X\subseteq S  \,\}\cup \{\ast, \diamond\}$ as in Lemma~\ref{0707a} and $S'= S\cup T$. Then $S'$ is a fixed background set extending $S$. Suppose $f$ is an $rs$ function over $S$.
Let $\mathcal{C}$ and $\mathcal{D}$ be as in Lemma~\ref{0707a}. Consider the reaction system $\mathcal{A}= (S', C\cup D) $ over $S'$. Clearly, $\mathcal{A}$ is minimal.

Suppose $X\subseteq S$.
By Lemma~\ref{2612a},
$$\res_{\mathcal{A}}(X)
=\res_{\mathcal{C}}(X \cap  (S \cup \{\diamond\}  ) ) \cup
\res_{\mathcal{D}}(X \cap T )= \res_{\mathcal{C}}(X  ) \cup 
 \res_{\mathcal{D}}(\emptyset)=  \res_{\mathcal{C}}(X  ).$$
Hence, by Lemma~\ref{2612a} again,
$$\res_{\mathcal{A}  }^2(X)
=\res_{\mathcal{A}  } (\res_{\mathcal{C}  }(X)   )=
\res_{\mathcal{C}}( \res_{\mathcal{C}}(X)  \cap (S\cup \{\diamond\}   ) )
\cup  \res_{\mathcal{D}}( \res_{\mathcal{C}}(X)  \cap T    ).$$
Note that $\res_{\mathcal{C}}(X)  \cap (S\cup \{\diamond\} )$ equals $\{\diamond\}$ because $\res_{\mathcal{C}}(X)\subseteq T$ and $\diamond \in \res_{\mathcal{C}}(X)$. It follows that
$\res_{\mathcal{C}}( \res_{\mathcal{C}}(X)  \cap (S\cup \{\diamond\}   )  )=\res_{\mathcal{C}}( \{\diamond\})=\emptyset$. Therefore,
$\res_{\mathcal{A}}^2(X)
=\res_{\mathcal{D}}( \res_{\mathcal{C}}(X)  )=  f(X)$ by Lemma~\ref{0707a}.
\end{proof}

Finally, we address a question related to Lemma~\ref{060118a}.
When $\vert S'\vert=\vert S\vert+l$,
the lemma identifies certain $rs$ functions that cannot be $k$-simulated by any reaction system over $S'$ whenever $k> \frac{2^{\vert S\vert +l}-2}{2^{\vert S\vert}-2} >2^l$. Does any of those $rs$ functions behave identically for some $k\leq 2^l$? The following theorem answers this negatively.

\begin{theorem}
	Suppose $\vert S'\vert =\vert S\vert+l$ for a nonnegative integer $l$. Then every $rs$ function over $S$ can be strongly $k$-simulated by some reaction system over $S'$ whenever $k\leq 2^l$.
\end{theorem}

\begin{proof}
	Suppose $f$ is an arbitrary $rs$ function over $S$. Since $f$ can be canonically specified by a unique maximally inhibited reaction system over $S$, the case $k=1$ is trivial.	
	Suppose $1<k\leq 2^l$. Let $T=S'\backslash S$.
 Note that $\vert 2^{T}\vert =2^l\geq k$.
	Let $\emptyset =L_1, L_2, \dotsc, L_k=T$ be any $k$ distinct subsets of $T$. 
	Consider the reaction system $\mathcal{B}=(T,B)$, where 
	\begin{multline*}
	B= \{\, (X\cup T, S\backslash X, f(X))\mid X\subseteq S\,\} \,\cup \\
	\{\, (L_i, T\backslash L_i, L_{i+1})\mid 1\leq i\leq k-1   \,\}\cup
	\bigcup_{t\in T} \{ \,(\{s\}, \{t\}, \{s\})\mid s\in S \, \}.
	\end{multline*}
	We claim that $f$ can be strongly $k$-simulated by $\mathcal{B}=(S',B)$.
	
	Suppose $X$ is an arbitrary subset of $S$. It suffices to show that $\res_{\mathcal{B}}^i (X)= X\cup L_{i+1}$ for all $0\leq i\leq k-1$ because then $\res_{\mathcal{B}}^k (X)=\res_{\mathcal{B}}(\res_{\mathcal{B}}^{k-1}(X) )= \res_{\mathcal{B}}( X\cup T)=    f(X)$.
	We argue by induction. Trivially, 
	$\res_{\mathcal{B}}^0(X)=X= X\cup L_1$. 
	For the induction step, suppose  $1\leq i\leq k-1$. Then  
	$\res_{\mathcal{B}}^{i} (X)=\res_{\mathcal{B}}(\res_{\mathcal{B}}^{i-1}(X) )= \res_{\mathcal{B}}(X\cup L_{i}  ) $ by the induction hypothesis. Since $L_{i}\neq T$, it follows that
$\res_{\mathcal{B}}(X\cup L_{i}  )= X\cup L_{i+1}$.	
\end{proof}

\section{Conclusion and Open Problems}

The \emph{reaction system rank} of any $rs$ function $f$ over $S$ is the smallest possible size of a set $A$ of reactions such that  
$f$ can be specified by the reaction system $(S,A)$.
Through Theorem~\ref{0403a}, it can be shown that the number of extra resources needed to simulate a given $rs$ function $f$ by some strictly minimal reaction systems is bounded by the reaction system rank of $f$. However, since the upper bound $2^{\vert S\vert}$ for reaction system rank is effectively attainable by $rs$ functions over $S$ (see \cite{teh2017irreducible}), in view of Theorem~\ref{0407b} and Theorem~\ref{120219}, it is intriguing but not surprising if no fixed $S'\supseteq S$ exists with $\vert S' \vert<\vert S\vert+2^{\vert S\vert}$ such that every $rs$ function over $S$ can be $2$-simulated by some strictly minimal reaction system 
 over $S'$. 

In another direction, one can study the difference between $k$-simulation and strong $k$-simulation. With respect to Theorem~\ref{0603a},
one can ask whether 
the class of simulating reaction systems can be further restricted to the ones that are strictly minimal. Additionally, would an extended background set $S'$ of size $\vert S\vert+2^{\vert S\vert}$ be sufficient to strongly $2$-simulate every $rs$ function over $S$ by some minimal reaction system over $S'$? 
If either question has a negative answer, then this would mean that strong $k$-simulation is essentially weaker than $k$-simulation in terms of generative power.

As a conclusion, simulation of reaction systems and its strong version
can be further studied and compared from the following perspectives:
\begin{enumerate}
\item the complexity of the simulating reaction system;
\item the relative size of the extended background set;
\item the order of $k$-simulation, that is, the value of $k$.
\end{enumerate}

Finally, every hybrid reaction system over $(S,T)$ can be viewed as a reaction system over $S\cup T$. Hence, it is not our intention to generalize the study of reaction systems by
 introducing hybrid reaction systems. It is simply natural and convenient to do so in this study.

\section*{Acknowledgment}
This work is an extension of that published in the proceedings paper \cite{teh2018note}. The first author acknowledges support of Fundamental Research Grant Scheme \linebreak  No.~203.PMATHS.6711644 of Ministry of Education, Malaysia, and Universiti Sains Malaysia. Furthermore, this work is completed during his sabbatical leave from \mbox{15 Nov 2018} to 14 Aug 2019, supported by Universiti Sains  Malaysia.

\vspace{5mm}



\end{document}